\documentclass[12pt,reqno]{amsart}
\newcommand{\cP}{\mathcal{P}}
\newcommand{\cD}{\mathcal{D}}

\newcommand{\cU}{\mathcal{U}}

\newcommand{\cW}{\mathcal{W}}

\newcommand{\lspan}{\operatorname{span}}
\newcommand{\sgn}{\operatorname{sgn}}
\newcommand{\rL}{\mathrm{L}}

\newcommand{\hW}{{\hat{W}}}
\newcommand{\Td}[1]{T_{#1}}
\newcommand{\Ad}[1]{A_{#1}}
\newcommand{\Bd}[1]{B_{#1}}
\newcommand{\Pud}[2]{P^{(#1)}_{#2}}
\newcommand{\Pab}[1]{\Pud{\al,\be}{#1}}
\newcommand{\hPab}{\hat{P}^{(\alpha,\beta)}}
\newcommand{\hPud}[2]{\hat{P}^{(#1)}_{#2}}
\newcommand{\hLk}{\hat{L}_{k,}{}}

\newcommand{\hmuk}{{\hat{\mu}_{k}}}
\newcommand{\hmuab}{{\hat{\mu}_{\alpha,\beta}}}
\newcommand{\Wab}{{W_{\alpha,\beta}}}
\newcommand{\hWab}{{\hat{W}_{\alpha,\beta}}}

\newcommand{\hphi}{{\hat{\phi}}}

\newcommand{\hA}{{\hat{A}}}
\newcommand{\hB}{{\hat{B}}}
\newcommand{\hT}{{\hat{T}}}
\newcommand{\hy}{{\hat{y}}}
\newcommand{\hq}{{\hat{q}}}
\newcommand{\hr}{{\hat{r}}}
\newcommand{\hb}{{\hat{b}}}
\newcommand{\hw}{{\hat{w}}}

\newcommand{\Rset}{\mathbb{R}}

\newcommand{\al}{\alpha}
\newcommand{\be}{\beta}

\newtheorem{thm}{Theorem}[section]

\newtheorem{prop}{Proposition}[section]
\newtheorem{lem}{Lemma}[section]

\theoremstyle{definition}
\newtheorem{example}{Example}[section]
\newtheorem{definition}{Definition}[section]

\theoremstyle{remark}

\begin{document}

\title[orthogonal polynomials and flags]{On orthogonal polynomials
  spanning a non-standard flag}

\author{David G\'omez-Ullate} \address{ Departamento de F\'isica
  Te\'orica II, Universidad Complutense de Madrid, 28040 Madrid,
  Spain} 

\author{ Niky Kamran } \address{Department of Mathematics and
  Statistics, McGill University Montreal, QC, H3A 2K6, Canada}

\author{Robert Milson} \address{Department of Mathematics and
  Statistics, Dalhousie University, Halifax, NS, B3H 3J5, Canada}

\keywords{Orthogonal polynomials; Sturm--Liouville problems;
  Exceptional polynomial subspaces}
\subjclass[2000]{Primary: 34L10, 42C05; Secondary:  33C45, 34B24}

\begin{abstract}
  We survey some recent developments in the theory of orthogonal
  polynomials defined by differential equations.  The key finding is
  that there exist orthogonal polynomials defined by 2nd order
  differential equations that fall outside the classical families of
  Jacobi, Laguerre, and Hermite polynomials.  Unlike the classical
  families, these new examples, called exceptional orthogonal
  polynomials, feature non-standard polynomial flags; the lowest
  degree polynomial has degree $m>0$.  In this paper we review the
  classification of codimension $m=1$ exceptional polynomials, and
  give a novel, compact proof of the fundamental classification
  theorem for codimension 1 polynomial flags.  As well, we describe
  the mechanism or rational factorizations of 2nd order operators as
  the analogue of the Darboux transformation in this context.  We
  finish with the example of higher codimension generalization of
  Jacobi polynomials and perform the complete analysis of parameter
  values for which these families have non-singular weights.
\end{abstract}
\maketitle
\section{Introduction}\label{sec:Intro}
Even though the orthogonal polynomials of Hermite, Laguerre and Jacobi
arose from various applications in applied mathematics and physics,
these three families now serve as the foundational examples of
orthogonal polynomial theory.  As such, these classical examples admit
many interesting and fruitful generalizations.

A key property of classical orthogonal polynomials is
the fact that they can be defined by means of a Sturm-Liouville
problem.  One of the foundational results in this area a Theorem of
Bochner \cite{Bo} which states that if an infinite sequence of
polynomials of degree $0,1,2,\ldots$ satisfies a second order
eigenvalue equation of the form
\begin{equation}\label{eq:bochner}
  p(x)y'' + q(x) y' + r(x) y=\lambda y
\end{equation}
then $p(x),q(x)$ and $r(x)$ must be polynomials of degree at most
$2,1$ and $0$ respectively.  In addition, if the eigenpolynomial
sequence is $L^2$-orthogonal relative to a measure with finite
moments, then it has to be, up to an affine transformation of $x$, one
of the three classical classical families listed above
\cite{Aczel,Mikolas,Feldmann,Lesky,KL97}

Implicit in the above definition of classical polynomials is the
assumption that orthogonal polynomials form a basis for the standard
polynomial flag $\cP_0 \subset \cP_1\subset \cP_2\subset \ldots$,
where $\cP_n=\lspan\{ 1,x,x^2,\ldots, x^n\}$. In a pair of recent
papers \cite{GKM09a,GKM09b}, we showed that there exist complete
orthogonal polynomial systems defined by Sturm-Liouville problems that
extend beyond the above classical examples.  What distinguishes our
hypotheses from those made by Bochner is that the eigenfunction
corresponding to the lowest eigenpolynomial of the sequence need not
be of degree zero, even though the full set of eigenfunctions still
forms a basis of the weighted $\rL^2$ space.

Already some 20 years ago, particular examples of Hermite-like
polynomials with non-standard flags were described in the context of
supersymmetric quantum mechanics \cite{dubov,CKS95}.  The last few years have
seen a great deal of activity in the area of non-standard flags; the
topic now goes under the rubric of \emph{exceptional} orthogonal
polynomials.  There are applications to shape-invariant
potentials \cite{Quesne1}, supersymmetric transformations \cite{GKM10},
to discrete quantum mechanics \cite{SO2}, mass-dependent potentials
\cite{Midya-Roy}, and to quasi-exact solvability \cite{Tanaka}.  As
well, there are now examples of orthogonal polynomials that are
solutions of 2nd order equations and that form flags of arbitrarily
high codimension \cite{SO1}.

In light of the activity surrounding orthogonal polynomials with
non-standard flags, we hope that it will be useful to summarize some
key results and findings, and to supply stand-alone proofs to some key
propositions. We note that the adjective ``exceptional''
was introduced in the context of our investigation of the equivalence
problem for polynomial subspaces \cite{GKM09a,CrooksMilson}.  

\section{Preliminaries}
\subsection{Polynomial flags}

Let $U_1\subset U_2\subset \ldots$, where $\dim U_k = k$ be a flag of
real, finite-dimensional, polynomial subspaces.  Let $n_k$ denote the
degree of $U_k$; that is $n_k$ is the maximum of the degrees of the polynomials $p\in
U_k$.  Let $\ell_k=n_k+1-k$ denote the codimension of $U_k$ in
$\cP_{n_k}$, where the $n+1$ dimensional polynomial vector spaces
\[\cP_n=\lspan\{ 1,x,x^2,\ldots, x^n\}\]
make up the standard polynomial flag: $\cP_0 \subset \cP_1\subset
\cP_2\subset \ldots$.  We say that $\{ p_k\}_{k=1}^\infty$ is a basis
of the flag if
\begin{equation}
  \label{eq:Ukspan}
  U_k = \lspan \{ p_1,\ldots, p_k \}
\end{equation}
Note that no generality is lost if we assume that $\deg p_k  = n_k$.
\begin{definition}
  We call $\cU=\{U_k\}_{k=1}^\infty$ a degree-regular flag if $n_k <
  n_{k+1}$ for all $k$.  Equivalently, a flag is degree-regular if it
  admits a basis $\{ p_k \}_{k=1}^\infty$ such that $\deg p_k < \deg p_{k+1}$ for
  all $k$.
\end{definition}
\begin{prop}
  A polynomial flag is degree-regular if and only if the codimension
  sequence obeys $\ell_k \leq \ell_{k+1}$ for all $k$.
\end{prop}
\noindent Henceforth, we will always assume that all flags under
discussion are degree-regular.
\begin{definition}
  We say that the codimension of a polynomial flag is
  \emph{semi-stable} if the codimension sequence $\ell_k$ admits an
  upper bound. If so, we call $\ell=\lim_k \ell_k$ the codimension of
  the flag.  We say that codimension is \emph{stable} if $\ell_k=\ell$
  is constant for all $k$.
\end{definition}
\noindent Likewise, whenever we speak of the codimension of a
polynomial flag, we always assume that the flag has a semi-stable
codimension.

Let $\cD_2(\cU)$ denote the vector space of 2nd order differential operators
with rational real-valued coefficients
\begin{equation}
  \label{eq:Tydef}
  T(y) = p(x) y'' + q(x) y' + r(x) y,\quad p,q,r\in \Rset(x) 
\end{equation}
that preserve $\cU$; i.e., $T(U_k) \subset U_k$ for all $k$.  If
$p,q,r$ are polynomials, we call $T$ a polynomial operator.
Equivalently, an operator is non-polynomial if and only if it has a
pole.
\begin{definition}
  We say that $\cU$ is \emph{imprimitive} if it admits a common
  factor; i.e., $\cU$ is spanned by $\{qp_k\}$ where $q(x)$ is a
  polynomial of degree $\geq 1$. Otherwise, we call $\cU$
  \emph{primitive}.
\end{definition}
\noindent At this juncture, it is important to state the following two
Propositions.
\begin{prop}
  \label{prop:cramer}
  Let $T(y) = py''+qy'+ry$ be a differential operator such that
  \[T(y_i) = g_i,\; i=1,2,3,\]
  where $y_i, g_i$ are polynomials with
  $y_1, y_2, y_3$ linearly independent.  Then, $p,q,r$ are rational
  functions with  the Wronskian
  \[W(y_1,y_2,y_3) = \det \begin{pmatrix}
    y''_1 & y'_1 & y_1\\
    y''_2 & y'_2 & y_2\\
    y''_3 & y'_3 & y_3
  \end{pmatrix} \] in the denominator.
\end{prop}
\begin{proof}
  It suffices to apply Cramer's rule to solve the linear system
  \[ \begin{pmatrix}
    y''_1 & y'_1 & y_1\\
    y''_2 & y'_2 & y_2\\
    y''_3 & y'_3 & y_3
  \end{pmatrix}
  \begin{pmatrix}
    p\\q\\r
  \end{pmatrix}
  =
  \begin{pmatrix}
    g_1\\ g_2\\ g_3
  \end{pmatrix}
  \]
\end{proof}
\noindent
We should also note that there is a natural linear isomorphism between the space $\cD_2(\cU)$ of a flag 
$\cU$ spanned by $\{p_k\}_{k=1}^{\infty}$ and the space $\cD_2(\tilde{\cU})$ of the imprimitive flag
$\tilde{\cU}$ spanned by $\{ q p_k \}_{k=1}^{\infty}$.
\begin{prop}
  \label{prop:gaugeequiv}
  Let $T(y)$ be an operator that preserves $\cU$. Then, the
  gauge-equivalent operator $\tilde{T} = q T q^{-1}$ preserves
  $\tilde{\cU}$.
\end{prop}
Proposition \ref{prop:cramer} makes clear why we restrict our
definition of $\cD_2$ to operators with rational coefficients.
Proposition \ref{prop:gaugeequiv} explains the need for primitive
flags; these are the canonical representatives for the equivalence
relation under gauge transformations.

\begin{definition}
  Let $\cU$ be a polynomial flag of semi-stable codimension.  We say
  that $\cU$ is an \emph{exceptional flag} if $\cU$ is primitive and
  if $\cD_2(\cU)$ does not preserve a polynomial flag of smaller
  codimension.
\end{definition}
\noindent 

Here are some examples to illustrate the above definitions.
\begin{example}
  The codimension 1 flag spanned by $ 1,x^2,x^3,\ldots,$ is
  exceptional because the non-polynomial operator
  \[T(y) = y''-2y'/x\] preserves the flag.  Since $\ell_1 = 0,
  \ell_k=1,\; k\geq 2$, the codimension of this flag is not stable,
  but only semi-stable.
\end{example}

\begin{example}
  By contrast, the flag spanned by $x+1, x^2, x^3,\ldots$ has a stable
  codimension $\ell=1$.  This flag is exceptional because it is
  preserved by the non-polynomial operator
  \[ T(y) =y''-2(1+1/x)y'+(2/x) y.\]
\end{example}

\begin{example}
  Let $h_k(x)$ denote the degree $k$ Hermite polynomial.  The
  codimension 1 flag spanned by $h_1, h_2, h_3,\ldots$ is not
  exceptional.  The flag is preserved by the operator $T(y) = y''-x
  y'$, but this operator also preserves the standard, codimension
  zero, polynomial flag.
\end{example}
\begin{example}
  The codimension 1 flag spanned by $x,x^2,x^3, \ldots$ is not
  exceptional because $x$ is a common factor. The flag is preserved by
  the operator
  \[ \tilde{T}(y) = y'' - 2y'/x+2y/x^2.\] However, observe that
  $\tilde{T} = x T x^{-1}$ where $T(y) = y''$ is an operator that
  preserves the standard, codimension zero, polynomial flag.
\end{example}

\begin{example}
  \label{ex:codim2}
  Let
  \begin{equation}
    \label{eq:pidef}
    y_{2k-1} = x^{2k-1} - (2k-1)x,\; y_{2k} = x^{2k} - k
    x^2,\quad k=2,3,4,\ldots
  \end{equation}
  Consider the codimension-2 flag spanned by $1,y_3, y_4,
  y_5,y_6,\ldots$. The degree sequence of the flag is
  $0,3,4,5,\ldots$; therefore the codimension is not stable, but
  merely semi-stable.  The flag is preserved by the following
  operators \cite{GKM04}:
  \begin{align}
    T_3(y) &= (x^2-1)y''-2x y',\\
    T_2(y) &= xy''-2(1+2/(x^2-1))y',\\
    T_1(y) &= y'' +x(1-4/(x^2-1))y'.
  \end{align}
  The flag is exceptional, because $T_1$ and $T_2$ do not preserve the
  standard, codimension zero flag.  These operators cannot preserve a
  codimension 1 flag, because, as will be shown in Lemma
  \ref{lem:1pole}, an operator that preserves a codimension 1 flag can
  have at most 1 pole.
\end{example}

\subsection{Orthogonal polynomials}
\begin{definition}
  We will say that a 2nd order operator $T(y)$ is exactly solvable by
  polynomials if the eigenvalue equation
  \begin{equation}
    \label{eq:Tylambday}
    T(y) = \lambda y.
  \end{equation}
  has infinitely many eigenpolynomial solutions $y=y_j$ with
  \[  \deg y_j < \deg y_{j+1},\quad j=1,2,\ldots.\]
\end{definition}

Let $I=(x_1, x_2)$ be an open interval (bounded, unbounded, or
semi-bounded) and let $W(x) dx$ be a positive measure on $I$ with
finite moments of all orders.  We say that a sequence of real
polynomials $\{y_j\}_{j=1}^\infty$ forms an \emph{orthogonal
  polynomial system} (OPS for short) if the polynomials constitute an
orthogonal basis of the Hilbert space $\rL^2(I,W dx)$.  If the
codimension of the corresponding flag is stable, or semi-stable, we
will say that the OPS has codimension $m$.

The following definition encapsulates the notion of a system of
orthogonal polynomials defined by a second-order differential
equation.  Consider the boundary value problem
\begin{align}
  \label{eq:bvp}
  &-(Py')'+ Ry  = \lambda W y\\
  \label{eq:pslpbc}
  &\lim_{x\to x_{i}^{\pm}} (Py'u-Pu'y)(x) = 0,\quad i=1,2,
\end{align}
where $P(x), W(x)>0$ on the interval $I=(x_1,x_2)$, and where $u(x)$
is a fixed polynomial solution of \eqref{eq:bvp}.  We speak of a
polynomial Sturm-Liouville problem (PSLP) if the resulting spectral
problem is self-adjoint, pure-point and if all eigenfunctions are
polynomial.  

If the eigenpolynomials span the standard flag, then we recover the
classical orthogonal polynomials, the totality of which is covered
by Bochner's theorem.  If the solution flag has a codimension $m>0$,
Bochner's theorem no longer applies and we encounter a generalized
class of polynomials; we name these exceptional, or $X_m$ polynomials.

Given a PSLP, the operator
\[ T(y) = W^{-1}(Py')' - W^{-1} Ry \] is exactly solvable by
polynomials.  Letting $p(x), q(x), r(x)$ be the rational coefficients
of $T(y)$ as in \eqref{eq:Tydef}, we have
\begin{align}
  \label{eq:Pdef}
  P(x) &= \exp\left(\int^x  \!\!q/p\,dx\right),\\
  \label{eq:Wdef}
  W(x) &= (P/p)(x),\\
  \label{eq:Rdef}
  R(x) &= -(rW)(x),
\end{align}
Hence, for a PSLP, $P(x), R(x), W(x)$ belong to the quasi-rational
class \cite{GB09}, meaning that their logarithmic derivative
is a rational function.

Conversely, given an operator $T(y)$ exactly solvable by polynomials and an interval $I=(x_1,x_2)$
we formulate a PLSP \eqref{eq:bvp} by employing
\eqref{eq:Pdef}--\eqref{eq:Rdef} as definitions, and by adjoining the
following assumptions:
\begin{itemize}
 \item[(PSLP1)] $P(x), W(x)$ are continuous and positive on $I$
 \item[(PSLP2)] $Wdx$ has finite moments: $\int_I x^n W(x) dx
   <\infty,\quad n\geq 0$
 \item[(PSLP3)] $ \lim_{x\to x_{i}} P(x) x^n = 0,\quad i=1,2,\quad
   n\geq 0$
 \item[(PSLP4)] the eigenpolynomials of $T(y)$ are dense in
   $\rL^2(I,Wdx)$.
\end{itemize}
These definitions and assumptions (PSLP1), (PSLP2) imply Green's
formula:
\begin{equation}
  \label{eq:Msymmetric}
  \int^{x_2}_{x_1}  T(f)g \, Wdx - \int_{x_1}^{x_1} T(g) f\, W dx=  P(f'g-fg') \Big|^{x_2}_{x_1}
\end{equation}
By (PSLP3) if $f(x),g(x)$ are polynomials, then the right-hand side is
zero.  If $f$ and $g$ are \emph{eigenpolynomials} of $T(y)$ with
unequal eigenvalues, then necessarily, they are orthogonal in
$\rL^2(I, Wdx)$.  Finally, by (PSLP4) the eigenpolynomials of $T(y)$
are complete in $\rL^2(I,Wdx)$, and hence satisfy the definition of an
OPS.
\section{Codimension 1 flags}
The key result in this paper is the following. An analogous theorem
for polynomial subspaces, rather than flags, was proved in
\cite{GKM09a}.
\begin{thm}
  \label{thm:main}
  Up to affine transformations, the flag spanned by $\{x+1,
  x^2,x^3,\ldots \}$ is the unique stable codimension 1 exceptional
  flag.
\end{thm}
\noindent
The proof proceeds by means of two lemmas.

\begin{lem}
  \label{lem:polyop}
  A primitive, codimension 1 polynomial flag is exceptional if and
  only if $\cD_2$ includes a non-polynomial operator.
\end{lem}
\begin{proof}
  It's clear that a non-polynomial operator cannot preserve the
  standard polynomial flag.  The standard flag is the unique
  codimension zero flag. Therefore, if a non-polynomial operator
  preserves a primitive, codimension 1 flag, that flag must be
  exceptional.

  Let us prove the converse. Let $T(y)$ be a polynomial operator that
  preserves a polynomial flag.  Consider the degree homogeneous
  decomposition
  \begin{equation}
    \label{eq:Tyhomdecomp}
    T(y) = \sum_{d=-2}^N T_d(y)     
  \end{equation}
  where
  \begin{align}
    \label{eq:Tdydef}
    T_d(y) &= x^d (\alpha_d x^2 y'' + \beta_d x y' +\gamma_d y)
  \end{align}
  and where $T_N$ is non-zero.  Since the operator is polynomial, we
  must have $\beta_{-2} = \gamma_{-2} = \gamma_{-1} = 0$. Also note that
  \begin{equation}
    \label{eq:deghomog}
    T_d(x^j) = (j(j-1) \alpha_d + \beta_d j + \gamma_d) x^{j+d}.
  \end{equation}
  Let 
  \[ y_k(x) = x^{n_k}+\text{lower deg. terms},\quad k=1,2,\ldots\] be
  a basis of the flag. If $N>0$, then the leading term $T_N(y)$ raises
  degree, and hence
  \[ T_N(x^{n_k}) = 0\] for all $k$.  By \eqref{eq:deghomog}, $T_N(y)$
  annihilates at most 2 distinct monomials, a contradiction.
  Therefore, $N\leq 0$, and we conclude that $T(y)$ is a
  Bochner-type operator \eqref{eq:bochner}.  However, such an
  operator preserves the standard polynomial flag.  Therefore, if
  $\cU$ is exceptional, then there is at least one operator in
  $\cD_2(\cU)$ that doesn't preserve the standard flag.  By the above
  argument, this operator must be non-polynomial.
\end{proof}

\begin{lem}
  \label{lem:1pole}
  Let $T(y)$ be a non-polynomial operator that leaves invariant a
  codimension 1 polynomial subspace $U\subset \cP_n$.  If $n\geq 5$,
  then $T(y)$ has exactly one pole.  Furthermore, up to a translation
  in $x$, a basis of $U$ has one of the following three forms:
  \begin{align}
    \label{eq:x1nu0}
    &x,x^2,x^3,\ldots, x^n\\ 
    \label{eq:x1nu1}
    &1+ax, x^2, x^3,\ldots, x^n\\
    \label{eq:x1nu2}
    &1+a_1 x^2,x+a_2x^2,x^3,\ldots, x^n
  \end{align}
\end{lem}
\begin{proof}
  By applying a translation if necessary, there is no loss of generality in assuming
  that $x=0$ is a pole of $T(y)$.  Since the codimension is 1, the
  subspace admits an order-reduced basis of the form
  \begin{align*}
    y_j &= x^{j-1}+a_{j-1} x^\nu, &&j=1,\ldots, \nu\\ \nonumber
    y_j &= x^{j},&&  j=\nu+1\ldots n,
  \end{align*}
  where $0\leq \nu \leq n$.  The matrix representation of this basis
  is a $n\times n$ matrix in row-reduced echelon form. This matrix has
  $n$ pivots and 1 gap in position $\nu$.

  Our first claim is that $\nu \leq 2$.  Suppose not.  Then 
  \[ W(y_1,y_2,y_3) = W(1,x,x^2) + \text{higher degree terms} = 2 +
  O(x).\] By assumption, $T(y_i)$ is a polynomial. Hence, Proposition
  \ref{prop:cramer} implies that $x=0$ is not a pole of the operator, a
  contradiction.  

  Having established that $\nu\leq 2$, we observe that $x^3,x^4,x^5\in
  U$.  The Wronskian of these monomials is a multiple of
  $x^9$. Therefore, by Proposition \ref{prop:cramer}, $x=0$ is the
  unique pole.
\end{proof}

We are now ready to give the proof of Theorem \ref{thm:main}.  Let
$\cU$ be an exceptional polynomial flag with stable codimension 1.  By
Lemma \ref{lem:polyop}, there exists a non-polynomial operator $T(y)$
that preserves the flag.  By Lemma \ref{lem:1pole}, this operator has
a unique pole.  Without loss of generality we assume that $x=0$ is the
unique pole.  We rule out possibility \eqref{eq:x1nu0}, because if
this holds for even one $U_k,\; k\geq 5$,  then it must hold for all
$k$. This would imply that $x,x^2,x^3,\ldots$ is a basis of the flag
--- a violation of the primitivity assumption.

Let us rule out possibility \eqref{eq:x1nu2}.  Let $T(y)$ be a
non-polynomial operator that preserves the flag.  Since $x=0$ is the
unique pole, we can decompose the operator into degree-homogeneous
terms
\[ T(y) = \sum_{d=N_1}^{N_2} T_d(y) \] where $T_d(y)$ has the form
shown in \eqref{eq:Tdydef}, and where $T_{N_1}$ and $T_{N_2}$ are
non-zero.  By the argument used in the proof of Lemma
\ref{lem:polyop}, we must have $N_2\leq 0$.  By assumption, $T(1+a_1
x^2), T(x+a_2 x^2), T(x^3)$ are polynomials.  However, if $N_1\leq
-4$, then this condition requires that 
\[T_{N_1}(1)= T_{N_1}(x) = T_{N_1}(x^3) = 0,\] which means that
$T_{N_1}$ is zero --- a contradiction.  If $N_1=-3$, then $T_{-3}$
annihilates $1,x,x^5$, a contradiction.  If $N_1=-2$, then $T_{-2}$
annihilates $1,x,x^4$, another contradiction.  Similarly, if $N_1=-1$,
then $T_{-1}$ annihilates 1.  Hence $T_{-1}(y) = \alpha_{-1}xy'' +
\beta_{-1}y'$.  This means that there is no pole --- a contradiction.
Therefore, $N_1=N_2=0$, but that means that $T(y)$ is a polynomial
operator --- again, a contradiction.  This rules out possibility
\eqref{eq:x1nu2}.

This leaves \eqref{eq:x1nu1} as the only possibility.  Since we assume
that the codimension is stable, $\deg(1+ax) = 1$ and hence $a\neq 0$.
We scale $x$ to transform $1+ax$ to $x+1$.  This concludes the proof
of the main theorem.

\subsection{$X_1$ polynomials}
The above theorem explains the origin of the adjective ``exceptional''
and leads directly to a more general class of orthogonal polynomials
--- outside the class described by Bochner's theorem.  Two families of
orthogonal polynomials arise naturally when we consider codimension 1
flags.  We describe these $X_1$ polynomials below.

In order to construct codimension 1 polynomial systems, we must
consider $\cD_2$ of the flag spanned by $x+1,x^2,x^3,\ldots$.  
\begin{prop}
  The most general 2nd order operator that preserves the flag $\{ x+1,
  x^2,x^3,\ldots\}$ has the form
  \begin{equation}
    \label{eq:T1genop}
    T(y) = (k_2 x^2+ k_1 x+k_0) y''-(x+1)\left(k_1
      +\frac{2k_0}{x}\right)y'+ \left(k_1 +
      \frac{2k_0}{x}\right)y,
  \end{equation}
  where $k_2,k_1,k_0$ are real constants.
\end{prop}
\noindent See \cite[Proposition 4.10]{GKM6} for the proof.

Thus to obtain $X_1$ orthogonal polynomials it suffices to determine
all possible values of $k_2, k_1, k_0$ for which $P(x), W(x)$ as given
by \eqref{eq:Pdef} \eqref{eq:Wdef} satisfy the conditions of a PSLP.
This analysis is performed in \cite{GKM09b}.  In summary, non-singular
weights arise only for the case where $k_2 x^2+k_1 x + k_0$ either has
two distinct real roots or if $k_2=0, k_1\neq 0$.  The first case
leads to the Jacobi $X_1$ polynomials; the second leads to the $X_1$
Laguerre polynomials. In both cases, the polynomial flags span a dense
subspace of the respective Hilbert space \cite[Proposition 3.1,
Proposition 3.3]{GKM09b}. We summarize the key properties of these two
families below.

\subsection{$X_1$-Jacobi polynomials}
Let $\alpha\neq \beta$ be real parameters such that
\begin{equation}
\label{alfabetacond}
\alpha>-1,\quad \beta>-1, \quad \sgn \alpha=\sgn\beta.
\end{equation}
Set
\begin{equation}\label{eq:abcdef}
  a= \frac{1}{2}(\beta-\alpha),\quad
  b= \frac{\beta+\alpha}{\beta-\alpha},\quad
  c=b+1/a.
\end{equation}
Note that, with the above assumptions, we have $|b|>1$.  Now let
\begin{equation}
  \label{eq:ui} u_1=x-c,\qquad u_i=(x-b)^i,\quad  i\geq 2.
\end{equation}
Define the measure $d\hmuab=\hWab\,dx$ where
\begin{equation}
  \label{eq:Wabdef}
 \hWab = \frac{(1-x)^\alpha  (1+x)^\beta}{(x-b)^2},\quad x\in (-1,1).
\end{equation}
Since $\hWab>0$ for $-1<x<1$, the scalar product
\begin{equation}
  \label{eq:jacobiproduct}
   (f,g)_{\alpha,\beta} := \int^1_{-1} f(x) g(x)\,  d\hmuab,
\end{equation}
is positive definite. Also note that the above measure has finite
moments of all orders.

We now define the $X_1$-Jacobi polynomials $\hPab_n,\, n=1,2,\ldots$
as the sequence obtained by orthogonalization of the flag spanned by
$u_1, u_2, u_3, \ldots$ with respect to the scalar product
\eqref{eq:jacobiproduct}, and by imposing the normalization
condition\footnote{This convention differs from the one adopted in
  \cite{GKM09b}.  The change is made to conform with the convention
  adopted below for the generalized $X_m$ Jacobi polynomials.}
\begin{equation}
  \label{eq:jacobinorm}
  \hPab_n(1) = n\binom{\alpha+n-1}{n}.
\end{equation}
\noindent
From their definition it is obvious that $\deg \hPab_n=n$. However, as
opposed to the ordinary Jacobi polynomials, the sequence starts with a
degree one polynomial.

Next, define the 2nd order operator
\begin{equation}
  \label{eq:Tabdef}
  \Td{\al,\be}(y) = (x^2-1) y''+2a
  \left(\frac{1-b\,x}{b-x}\right) \big((x-c)y'-y\big),
\end{equation}
An elementary calculation shows that this operator preserves the flag
$\{u_i\}$ as defined above.  Indeed, the flag $\{u_i\}$ and the
operator \eqref{eq:Tabdef} are obtained from the flag $x+1,
x^2,x^3,\ldots$ and from the operator \eqref{eq:T1genop} via the
following specialization and affine transformation
\begin{equation}
  \label{eq:X1jacparams}
  k_2 = 1, \quad k_1 = -2ab,\quad k_0 = (1-b^2)a, \quad x \to -a(x-b).
\end{equation}
Multiplying both sides of the equation
\[ -\Td{\al,\be}(y) =\lambda y \] by $\hWab$ leads to the
following PSLP:
\begin{align}
  \label{eq:JacobiSLP}
  & ((1-x^2) \hWab\, y')' + 2a
  \left(\frac{1-b\,x}{b-x}\right) \hWab\, y   = \lambda\,  \hWab\, y,\\
  \label{eq:JacobiSLPa}
  &\lim_{x\to 1^-} (1-x)^{\al+1} ((x-c)y'-y) = 0,\\
  \label{eq:JacobiSLPb}
  &\lim_{x\to -1^+} (1+x)^{\be+1} ((x-c)y'-y) = 0.
\end{align}
The boundary conditions select polynomial solutions.  The self-adjoint
form of \eqref{eq:JacobiSLP} ensures that the solutions are orthogonal
relative to $d\hmuab$.  Therefore, the $X_1$ Jacobi polynomials can
also be described as polynomial solutions, $y = \hPab_n(x)$, of the
following 2nd order equation:
\begin{equation}
  \label{eq:X1Jacobilambda}
   \Td{\al,\be}(y) = (n-1)(\al+\be+n) y.  
\end{equation}
\subsection{$X_1$-Laguerre polynomials}

For $k>0$, set
\begin{equation}
  \label{eq:vi}
  v_1=x+k+1,\qquad v_i=(x+k)^i,\quad
  i\geq 2
\end{equation}
Define the measure $d\hmuk= \hat W_k\,dx$ where
\begin{equation}
  \label{eq:Wkdef}
  \hat W_k = \frac{e^{-x} x^k}{(x+k)^2},\quad x\in (0,\infty)
\end{equation}
and observe that $\hat W_k>0$ on the domain in question.  Therefore,
the following inner product is positive definite:
\begin{equation}
  \label{eq:Laguerreproduct}
   (f,g)_k:= \int^\infty_{0} f(x) g(x)\,d\hmuk,
\end{equation}
Also note that the above measure has finite moments of all orders.

We define the $X_1$-Laguerre polynomials $\hLk_i,\; i=1,2,3,\ldots$ as
the sequence obtained by orthogonalization of the flag spanned by
$v_1, v_2, v_3,\ldots$ with respect to the scalar product
\eqref{eq:Laguerreproduct} and subject to the normalization condition
\begin{equation}
  \label{eq:lagnorm}
  \hLk_n(x) = \frac{(-1)^nx^n}{(n-1)!} + \text{ lower order terms}\quad n\geq 1.
\end{equation}
Again, since we are orthogonalizing a non-standard flag, the
$X_{1}$-Laguerre polynomial sequence starts with a polynomial of
degree $1$, rather than a polynomial of degree $0$.

Define the operator
\begin{equation}
  \label{eq:Tkdef}
  T_k(y)= -x   y''+\left(\frac{x-k}{x+k}\right)\big((x+k+1)y'-y\big)
\end{equation}
and note that this operator leaves invariant the flag spanned by the
$v_i$.  Indeed the flag $\{v_i\}$ and \eqref{eq:Tkdef} are obtained
from the flag $x+1,x^2,\ldots$ and \eqref{eq:T1genop} via the
specializations and an affine transformation shown below:
\begin{equation}
  k_2 =0,\; k_1=-1,\; k_0 = k,\quad x\mapsto x+k
\end{equation}
The corresponding PSLP takes the form
\begin{align}
  \label{eq:LaguerreX1SLP}
  &(x \hW_k y')' + \left(\frac{x-k}{x+k}\right) \hW_k y = \lambda y,\\
  &\lim_{x\to 0^+} x^{k+1}  ((x+k+1) y'-y)=0,\\
  & \lim_{x\to \infty} e^{-x} ((x+k+1)y'-y) = 0.
\end{align}
As before, the boundary conditions select polynomial solutions, while
the self-adjoint form of \eqref{eq:LaguerreX1SLP} ensures that these
solutions are orthogonal relative to $d\hmuk$.  Therefore, the $X_1$
Laguerre polynomials can also be defined as polynomial solutions,
$y=\hLk_n$, of the following 2nd order equation:
\begin{equation}
  \label{eq:X1Laguerrelambda}
  T_k(y) = (n-1) y.
\end{equation}

\noindent
Having introduced the $X_1$ Jacobi and Laguerre polynomials, we are
able to state the following corollary of Theorem \ref{thm:main}.  The
proof is found in \cite{GKM09b}.
\begin{thm}
  \label{thm:main1}
  The $X_1$-Jacobi polynomials and the $X_1$ Laguerre polynomials are
  the unique orthogonal polynomial systems defined by a stable
  codimension-1 PSLP.
\end{thm}

\section{Higher codimension flags}
Even though the first examples of exceptional orthogonal polynomials
involve codimension-1 flags, recently announced examples \cite{SO1}
are proof that exceptional orthogonal polynomials can span flags of
arbitrarily high codimension.  Another important development is the
recent proof \cite{Quesne1} that the codimension-1 families can be
related to the classical orthogonal polynomials by means of a Darboux
transformation.  In a follow-up publication \cite{GKM10}, it was shown
that the new higher codimension examples are systematically derivable
by means of algebraic Darboux transformations \cite{GKM04}.

A closely related development involves the notion of shape-invariance
\cite{GEND83}, a methodology related to the study of exactly solvable
potentials.  The close connection between solvable potentials and
orthogonal polynomials is well recognized; consider the relationship
between the harmonic oscillator and Hermite polynomials, for example.

In the orthogonal polynomial context, we have to factorize general
second order operators, not just Schr\"odinger operators.  The
shape-invariance property of the classical potentials manifests as the
Rodrigues' formula for the corresponding polynomials.  It has been
pointed out that all the potentials related to exceptional orthogonal
polynomials exhibit the shape-invariance property\cite{Quesne1,SO4},
and that therefore, like their classical counterparts, exceptional
orthogonal polynomials have a Rodrigues' formula.  This phenomenon was
studied and explained in \cite{GKM10}, where it was shown that the
shape-invariance property of the $X_m$ (exceptional, codimension $m$)
polynomials follows from the permutability property of higher-order
Darboux transformations.

\subsection{The Darboux transformation}
In the remainder of this section, we review the some key definitions and
results related to Darboux transformations and shape-invariance, and
then illustrate these ideas with the example of $X_m$ Jacobi
polynomials\cite{STZ10}.

Consider the differential operators:
\begin{align}
  T(y)&=py''+q y'+ry\\
  \label{eq:ABdef}
  A(y) &= b(y'-wy),\quad  B(y) = \hb(y'-\hw y),
\end{align}
where $p,q,r,b,w,\hb,\hw$ are rational functions.
\begin{definition}
  We speak of a rational factorization if  there exists a constant
  $\lambda_0$ such that  
  \begin{equation}
    \label{eq:TBA}
    T = B A+ \lambda_0
  \end{equation}
  If the above equation holds, we call 
  \begin{equation}
    \label{eq:hTAB}
    \hT= A B  +\lambda_0.
  \end{equation}
  the partner operator.  We call
  \begin{equation}
    \label{eq:phidef}
    \phi(x) = \exp\int^x\!\! w\,dx,\quad w=\phi'/\phi
  \end{equation}
  a quasi-rational factorization eigenfunction and $b(x)$ the
  factorization gauge.
\end{definition}
The reason for the above terminology is as follows.  By \eqref{eq:TBA},
\begin{equation}
  \label{eq:phi0rel}
  T(\phi) = \lambda_0 \phi;
\end{equation}
hence the term factorization eigenfunction.  Next, consider two
factorization gauges $b_1(x), b_2(x)$ and let $\hT_1(y), \hT_2(y)$ be
the corresponding partner operators.  Then,
\[ \hT_2 = \mu^{-1} \hT_1\mu,\quad \text{where } \mu(x) = b_1(x)/
b_2(x).\] Therefore, the choice of $b(x)$ determines the gauge of the
partner operator. This is why we refer to $b(x)$ as the factorization
gauge.  Also, note that in \cite{GKM04} the above construction was
referred to as an \emph{algebraic Darboux transformation}.  However,
in light of the recently recognized role played by operators with
rational coefficients, the term \emph{rational factorization} seems to
be preferable.

\begin{prop}
  Let $T(y)$ be a 2nd order operator exactly solvable by polynomials,
  and let $\phi(x)$ be a quasi-rational factorization eigenfunction with
  eigenvalue $\lambda_0$.  Then, there exists a rational factorization
  \eqref{eq:TBA} such that the partner operator is also exactly
  solvable by polynomials, and such that the partner flag is primitive
  (no common factors).
\end{prop}
\begin{proof}
  Let $w(x)=\phi'(x)/\phi(x)$ and let $b(x)$ be an as yet unspecified
  rational function. Set
  \begin{align}
    \label{eq:hwdef}
    &\hw = -w-q/p+b'/b,\\
    &\hb= p/b,
  \end{align}
  and let $A(y), B(y)$ be as shown in in \eqref{eq:ABdef}.  An
  elementary calculation shows that \eqref{eq:TBA} holds.  Let $y_1,
  y_2, \ldots$ be a degree-regular basis of the eigenpolynomials of
  $T(y)$.  We require that the flag spanned by $A(y_j)$ be polynomial
  and primitive (no common factors). Observe that if we take $b(x)$ to
  be the reduced denominator of $w(x)$, then $A(y_j)$ is a polynomial
  for all $j$.  However, this does not guarantee that $A(y_j)$ is free
  of a common factor.  That is a stronger condition, one that fixes
  $b(x)$ up to a choice of scalar multiple.  Finally, the intertwining
  relation
  \begin{equation}
    \label{eq:hTAAT}
    \hT A = A T
  \end{equation}
  implies that the $A(y_j)$ are eigenpolynomials of the partner $\hT$.
\end{proof}

Finally, let us derive the formula for the partner weight function.
\begin{prop}
  \label{prop:dualbphiW}
  Suppose that a PSLP operators $T(y)=py''+qy'+ry$ is related to a
  PSLP operator $\hT(y)=py''+\hq y'+ \hr y$ by a rational
  factorization with factorization eigenfunction $\phi(x)$ and
  factorization gauge $b(x)$, Then the dual factorization gauge,
  factorization eigenfunction and weight function are given by
  \begin{align}
    \label{eq:hbdef}
    & b\hb =p \\
    \label{eq:hWdef}
    &\hW/\hb =  W/b,\\
    \label{eq:hphidef}
    &\hb\hphi=1/(W\phi)
  \end{align}
\end{prop}
\begin{proof}
  Equation \eqref{eq:hbdef} follows immediately from \eqref{eq:ABdef}
  \eqref{eq:TBA}.  Writing
  \begin{equation}
    \label{eq:hTdef}
    \hT(y) = p y'' + \hq y' + \hr y,
  \end{equation}
  equation \eqref{eq:hTAB} implies that
  \begin{equation}
    \label{eq:whwrel}
    w+\hw = -q/p+b'/b = -\hq/p + \hb'/\hb.
  \end{equation}
  Hence,
  \begin{equation}
    \label{eq:hqqrel}
    \hq = q+p'-2pb'/b.
  \end{equation}
  From here, \eqref{eq:hWdef} follows by equations \eqref{eq:Pdef}
  \eqref{eq:Wdef}.  Equation \eqref{eq:hphidef} follows from
  \eqref{eq:phidef}.
\end{proof}

The adjoint relation between $A$ and $B$ allows us to compare the
$\rL^2$ norms of the two families.  
\begin{prop}
  Let $T(y), \hT(y)$ be PSLP operators related by a rational
  factorization \eqref{eq:TBA} \eqref{eq:hTAB}.  Let $\{y_j\}$ be the
  eigenpolynomials of $T(y)$ and let $\hy_j = A(y_j)$ be the
  corresponding partner eigenpolynomials. Then
  \begin{equation}
    \label{eq:Ayjnorm}
    \int^{x_2}_{x_1} \!\!A(y_j)^2 \, \hW dx  =  (\lambda_0-\lambda_j)
    \int^{x_2}_{x_1}   y_j^2 \, Wdx  
  \end{equation}
  where $x_1<x_2$ are the end points of the Sturm-Liouville problem
  in question.
\end{prop}
\begin{proof}
  As a consequence of \eqref{eq:hWdef}, $A$ and $-B$ are formally
  adjoint relative to the respective measures:
  \begin{equation}
    \label{eq:ABadj}
    \int^{x_2}_{x_1}  A(f) g\,\hW dx +
    \int^{x_2}_{x_1}  B(g) f\, Wdx  =  (P/b) f g \Big|^{x_2}_{x_1},
  \end{equation}
  where $P(x)$ is defined by \eqref{eq:Pdef} and where $b(x)$ is the
  factorization gauge.  By assumption both $W(x), \hW(x)$ are positive
  on $(x_1, x_2)$.  By \eqref{eq:hWdef}, $\hW = P/b^2$. In particular,
  the numerator $b(x)$ cannot have any zeroes in $(x_1, x_2)$.  As
  well, if either $x_1, x_2$ are finite, we cannot have $b(x_i)=0$,
  because that would imply that $\hW$ is not square integrable near
  $x=x_i$.  Hence $|1/b(x)|$ is bounded from above on $(x_1,x_2)$.
  Therefore, if $f,g$ are polynomials then the right-hand side of
  \eqref{eq:ABadj} vanishes by (PSLP3).  Therefore,
  \begin{equation}
    \int^{x_2}_{x_1} \!\!\!A(y_j)^2 \, \hW dx  = - \int^{x_2}_{x_1}\!\!
    B(A(y_j)) y_j\,  W dx = 
    (\lambda_0-\lambda_j) \int^{x_2}_{x_1}  \!\! y_j^2 \, Wdx   
  \end{equation}
\end{proof}

\subsection{Shape-invariance}
Parallel to the $L^2$ spectral theory \cite{deift,sukumar}, rational
factorizations of a solvable operator $T(y)$ can be categorized as
formally state-deleting, formally state-adding, or formally
isospectral.  The connection between these formal, algebraic Darboux
transformations and their $L^2$ analogues is discussed in
\cite{GKM10}.  We speak of a formally state-deleting transformation if
the factorization eigenfunction $\phi(x)$ is the lowest degree
eigenpolynomial of $T(y)$.  We speak of a formally state-adding
transformation if the partner factorization eigenfunction
\[\hphi(x) = \exp \int^x \!\!{\hw}\,dx ,\]
with $\hw$ defined by \eqref{eq:hwdef}, is a polynomial.  We speak of
a formally isospectral transformation if neither $\phi$ nor $\hphi$
are polynomials.  Examples of all three types of factorizations will
be given below.

State-adding and state-deleting factorizations are
dual notions, in the sense that if the factorization of $T$ is
state-deleting, then the factorization of $\hT$ is state-adding, and
vice versa.

\begin{definition}
  Let $\kappa\in K$ be a parameter index set and let
  \begin{equation}
    \label{eq:Tkydef}
    T_\kappa(y) = p(x) y'' + q_\kappa(x) y' + r_\kappa(x) y,\quad \kappa\in K,
  \end{equation}
  be a family of operators that are exactly solvable by polynomials.
  If this family is closed with respect to the state-deleting
  transformation, we speak of \emph{shape-invariant} operators.
\end{definition}
To be more precise, let $\pi_\kappa(x)= y_{\kappa,1}(x)$ be be the
corresponding eigenpolynomial of lowest degree. Without loss of
generality, we assume that the corresponding spectral value is zero,
and let
\begin{equation}
  \label{eq:Tkfact}
  T_\kappa = B_\kappa A_\kappa, \qquad A_\kappa\pi_\kappa = 0
\end{equation}
be the corresponding factorization.  Shape-invariance means that there
exists a one-to-one map $h:K\to K$ and real constants $\lambda_\kappa$
such that
\begin{equation}
  \label{eq:Thkfact}
 T_{h(\kappa)}= A_\kappa B_\kappa +\lambda_\kappa.
\end{equation}

In accordance with \eqref{eq:Pdef}, define
\begin{equation}
  \label{eq:Pkdef}
  P_\kappa(x) = \exp\left(\int^x \!\!q_\kappa/p\right)
\end{equation}
Let $b_\kappa(x)$ denote the shape-invariant factorization gauge; i.e.;
\begin{equation}
  \label{eq:Akdef}
  A_\kappa(y) = (b_\kappa/\pi_\kappa) \cW(\pi_\kappa,y),
\end{equation}
where
\begin{equation}
  \cW(f,g) = f g' - f' g
\end{equation}
denotes the Wronskian operator.
By equation \eqref{eq:Thkfact},
\begin{equation}
  q_{h(\kappa)} = q_\kappa+p'-2p b_\kappa'/b_\kappa.
\end{equation}
It follows that
\begin{equation}
  \label{eq:bk2rel}
  p\,P_\kappa/P_{h(\kappa)}  = b_\kappa^2.
\end{equation}
Below, we will use this necessary condition to derive the
factorization gauge of a shape-invariant factorization.
\subsection{Jacobi polynomials}
The operator
\begin{equation}
  \label{eq:Lkydef}
  \Td{\al,\be}(y) := (1-x^2)y'' +
  (\be-\al+(\al+\be+2)x)y', 
\end{equation}
preserves the standard flag, and hence is exactly solvable by
polynomials.  The classical Jacobi polynomials $\Pab{n}(x),\;
\al,\be>-1,\; n=0,1,2,\ldots$ are the corresponding
eigenpolynomials:
\begin{equation}
  \label{eq:Lnkeigenvalue}
   \Td{\al,\be}\Pab{n} = -n(n+\al+\be+1)\Pab{n},
\end{equation}
subject to the normalization
\[ \Pab{n}(1) = \binom{n+\al}{n} \] The $L^2$ orthogonality is
relative to the measure $\Wab(x) dx$ where 
\begin{equation}
 \Wab(x)=(1-x)^\al (1+x)^\be,\quad x\in (-1,1).
\end{equation}

The classical operators are
shape-invariant by virtue of the following factorizations:
\begin{align}
  \label{eq:Tabsi1}
  \Td{\al,\be} &= \hB_{\al,\be} \hA_{\al,\be},\\ 
  \label{eq:Tabsi2}
  \Td{\al+1,\be+1} &= \hA_{\al,\be} \hB_{\al,\be} + \al+\be+2 \qquad \mbox{where}\\
  \hB_{\al,\be}(y) &= (1-x^2)y' + (\be-\al+(\al+\be+2)x)y,\\
  &= (1-x)^{-\al} (1+x)^{-\be}\left(y (1-x)^{\al+1} (1+x)^{\be+1}\right)'\\
  \hA_{\al,\be}(y) &= y'.
\end{align}
As a consequence, $\Bd{\al,\be}(y)$ acts as a raising operator:
\begin{equation}
  \label{eq:jacraising}
  \Bd{\al,\be} \Pud{\al+1,\be+1}{n} = -2(n+1) \Pab{n+1},
\end{equation}
and $\Ad{\al,\be}(y)$ as a lowering operator:
\begin{equation}
  \label{eq:jaclower}
  \Pab{n}\,'=\frac{1}{2}(1+\al+\be+n) P_{\al+1,\be+1,n-1}. \\
\end{equation}

The classical Rodrigues' formula, namely
\begin{equation}
  (1-x)^{-\al} (1-x)^{-\be}\frac{d^n}{dx^n}
  \left((1-x)^{\al+n} (1-x)^{\be+n}\right) = (-2)^n n! \Pab{n}(x),
\end{equation}
 follows by applying $n$ iterations
of the raising operator to the constant function.

The quasi-rational solutions of $\Td{\al,\be}(y) = \lambda y$ are known
\cite[Section 2.2]{bateman}:
\begin{align*}
  \phi_1(x)&= \Pab{m}(x), && \lambda_0 = -m(1+\al+\be+m)\\
  \phi_2(x)&=(1-x)^{-\al} (1+x)^{-\be} P_{-\al,-\be,m}(x) &&
  \lambda_0 =(1+m)(\al+\be-m)\\
  \phi_3(x)&= (1-x)^{-\al} P_{-\al,\be,m}(x) && \lambda_0 =
  (1+\be+m)(\al-m)\\
  \phi_4(x)&= (1-x)^{-\be} P_{\al,-\be,m}(x), && \lambda_0 = 
  (1+\al+m)(\be-m)
\end{align*}
where $m=0,1,2,\ldots$.  The corresponding factorizations were
analyzed in \cite{GKM04a}. Of these, $\phi_1$ with $m=0$ corresponds
to a state-deleting transformation and underlies the shape-invariance
of the classical Laguerre operator and the corresponding Rodrigues
formula.  For $m>0$, the $\phi_1$ factorization eigenfunctions yield
singular operators and hence do not yield novel orthogonal
polynomials.  The $\phi_2$ family results in a state-adding
transformation.  The resulting flags are semi-stable, like in Example
\ref{ex:codim2}; see \cite{GKM04} for a discussion.  The type $\phi_3$
$\phi_4$ factorizations result in novel orthogonal polynomials,
provided $\al,\be$ satisfy certain inequalities. These families
were referred to as the J1, J2 Jacobi polynomials in \cite{SO3}.  The
two families are related by the transformations $\al\leftrightarrow
\be, x \mapsto -x$. We therefore focus only on the $\phi_3$
factorization; no generality is lost.

The derivations that follow depend in an elementary fashion on the
following well-known identities of the Jacobi polynomials.  We will
apply them below without further comment.
\begin{align}
  \label{eq:jacid1}
  &\Pab{0}(x) = 1,\quad \Pab{n}(x) = 0,\quad n\leq -1,\\
  &\Pab{n}(x) = (-1)^n \Pab{n}(-x),\\
  & (x-1) \Pab{m}\,{}'(x) =
  (\al+m)\Pud{\al-1,\be+1}{m}(x) - \al \Pab{m}(x)\\
  &\Pab{n}\,'(x)=\frac{1}{2}(1+\al+\be+n) P_{\al+1,\be+1,n-1}(x), \\
  \label{eq:jacid5}
  & \Pud{\al,\be-1}{n}(x) - \Pud{\al-1,\be}{n}(x) = \Pab{n-1}(x)
\end{align}

\subsection{The $X_m$ Jacobi polynomials.}
Fix an integer $m\geq 1$ and $\al,\be>-1$, 
and set \begin{equation}
  \xi_{\al,\be,m}=\Pud{-\al,\be}{m}(x)
\end{equation}
Take $\phi_3(x)$ as the factorization eigenfunction and take 
\[ b(x) = (1-x) \xi_{\al,\be,m}\] as the factorization gauge.
Applying \eqref{eq:ABdef} \eqref{eq:hwdef} and the identities
\eqref{eq:jacid1}-\eqref{eq:jacid5} , we obtain the following rational
factorization of the Jacobi operator \eqref{eq:Tabdef}:
\begin{align}
  \label{eq:TabBAabm}
  &\Td{\al,\be} = \Bd{\al,\be,m} \Ad{\al,\be,m}
  -(m-\al)(m+\be+1)\quad\text{where} \\
  &\Ad{\al,\be,m}(y) = (1-x) \xi_{\al,\be,m}\,
  y'+(m-\al)\xi_{\al+1,\be+1,m}\, y\\
  &\Bd{\al,\be,m}(y) = ((1+x)y'+(1+\be)y)/\xi_{\al,\be,m}
\end{align}
By \eqref{eq:hphidef}, $\hphi(x)=(1+x)^{-1-\beta}$ is the dual
factorization eigenfunction.  Since neither $\phi_3(x)$ nor $\hphi(x)$
is a polynomial, \eqref{eq:TabBAabm} is an example of a formally
isospectral factorization.  The corresponding partner operator is
shown below:
\begin{align}
  T_{\al,\be,m} &= A_{\al+1,\be-1,m} B_{\al+1,\be-1,m}
  -(m-\al-1)(m+\beta) ,\\
  T_{\al,\be,m}(y) &= T_{\al,\be}(y)
  -2\rho_{\al+1,\be-1,m}((1-x^2)y'+b(1-x)\,y) \\\nonumber 
  &\qquad  +m(a-b-m+1)y,\quad\text{where}\\
  \label{eq:rhodef}
  \rho_{\al,\be,m} &= \xi_{\al,\be,m}'/\xi_{\al,\be,m}\\ 
  &= \frac{1}{2}(1-\al+\be+m)\,\xi_{\al-1,\be+1,m-1}/\xi_{\al,\be,m} 
\end{align}

Based on the above factorization, we define the $X_m$, exceptional
Jacobi polynomials to be
\begin{align}
  \label{eq:hPdef}
  &\hPud{\al,\be,m}{n} = \frac{(-1)^{m+1}}{\al+1+j} \,A_{\al+1,\be-1,m}
  \Pud{\al+1,\be-1}{j},\quad j=n-m\geq 0\\
  &\quad = (-1)^m\left[\frac{1+\al+\be+j}{2(\al+1+j)}\,(x-1)
    \Pud{-\al-1,\be-1}{m}
    \Pud{\al+2,\be}{j-1}\right.\\
  &\qquad\qquad\qquad\left. +
    \frac{1+\al-m}{\al+1+j}\,\Pud{-2-\al,\be}{m}
    \Pud{\al+1,\be-1}{j}\right]
\end{align}
By construction, these polynomials satisfy
\begin{equation}
  \label{eq:hPlambda}
  T_{\al,\be,m} \hPud{\al,\be,m}{n} = -(n-m)(1+\al+\be+n-m)
  \hPud{\al,\be,m}{n}
\end{equation}
With the above definition, the generalized Jacobi polynomials obey the
following normalization condition
\begin{equation}
  \label{eq:hPnorm}
  \hPud{\al,\be,m}{n}(1) = \binom{\al+n-m}{n} \binom{n}{m},\quad
  n\geq m.
\end{equation}
Note that the $X_m$ operators and polynomials extend the classic family:
\begin{align}
  &T_{\al,\be,0}(y) = T_{\al,\be}(y)\\
  &\hPud{\al,\be,0}{n} = \Pud{\al,\be}{n}.
\end{align}

The $L^2$ norms of the classical polynomials are given by
\begin{align*}
  \int_{-1}^1 \left[\Pud{\al,\be}{n}(x)\right]^2 (1-x)^\al (1+x)^\be
  dx = N^{\al,\be}_{n}\intertext{where}
  N^{\al,\be}_{n}=\frac{2^{\al+\be+1}\Gamma(\al+1+n) \Gamma(\be+1+n)}{n!
    (\al+\be+2n+1) \Gamma(\al+\be+n+1)}
\end{align*}
By \eqref{eq:hWdef}, the weight for the $X_m$ Jacobi polynomials is
given by
\begin{equation}
  \label{eq:hWjacdef}
  \hW_{\al,\be,m}(x) = \frac{(1-x)^\al (1+x)^\be}{\xi_{\al+1,\be-1,m}(x)^2}
\end{equation}
In order for $L^2$ orthogonality to hold for the generalized
polynomials, we restrict $\al,\be$ so that the denominator in the
above weight is non-zero for $-1< x< 1$.  We also want to avoid the
degenerate cases where $x=\pm 1$ is a root of $\xi_{\al+1,\be-1,m}$.
To ensure that we obtain a codimension $m$ flag, we also demand that
$\deg \xi_{\al+1,\be-1,m} = m$.  The proof of the following
Proposition follows from the analysis in \cite[Chapter 6.72]{Sz}.
\begin{prop}
  \label{prop:albe1}
  Suppose that $\al,\be>-1$. Then $\deg \xi_{\al+1,\be-1,m}=m$ and
  $\xi_{\al+1,\be-1,m}(\pm 1) \neq 0$ if and only if $\be \neq 0$ and
  \begin{equation}
    \label{eq:albenonsing}
    \al,\al-\be-m+1 \notin \{ 0,1,\ldots, m-1 \}
  \end{equation}
\end{prop}
\begin{prop}
  \label{prop:albe2}
  Suppose that the conditions of the preceding
  Proposition hold. The polynomial $\xi_{\al+1,\be-1,m}(x)$
  has no zeros in $(-1,1)$ if and only if  $\alpha>m-2$ and
  \[ \sgn(\al-m+1) = \sgn(\be).\]
\end{prop}
This is a good place to compare the above results to the parameter
inequalities imposed in \cite{SO3,STZ10}.  These references impose the
condition
\[ \al>\be>m-1/2.\] Unlike Proposition \ref{prop:albe2}, this
condition fails to describe the most general non-singular weight
$\hW_{\al,\be,m}$.  Consider the following examples:
\begin{align}
  \hW_{1/3,-1/2,2} &= 288^2
  \frac{(1-x)^{1/3}(1+x)^{-1/2}}{(7x^2+2x-41)^2}\\
  \hW_{5/4,1/2,2} &= 128^2 \frac{(1-x)^{5/4} (1+x)^{1/2}}{(5x^2-14x+29)^2}
\end{align}
Neither of the above examples satisfy the parameter inequalities of
\cite{SO3, STZ10}, but both weights are non-singular on $(-1,1)$ and
have finite moments of all orders.  On the other hand, the parameter
values $m=2, \alpha=3/2, \beta=1/2$ give
\begin{align*}
  &\hW_{3/2,1/2,2}(x) = \frac{3}{8} (1-x)^{3/2}(1+x)^{1/2},\\
  &\hPud{3/2,1/2,2}{2+k} = \frac{3}{8} \Pud{3/2,1/2}{k},\quad k\geq 0.  
\end{align*}
In other words, for certain singular values of the parameters, the
codimension is actually less than $m$, and in some instances (such as
the one above) even yield the classical polynomials.  The condition
\eqref{eq:albenonsing} must be imposed in order to avoid such singular
possibilities.

\begin{prop}
  Suppose that $\al,\be>-1$ satisfy the conditions of Propositions
  \ref{prop:albe1} and \ref{prop:albe2}
  The $L^2$ norms of the $X_m$ Jacobi polynomials are given by
  \[
  \int_{-1}^1 \left[\hPud{\al,\be,m}{m+k}(x)\right]^2 \hW_{\al,\be,m} dx =
  \frac{(1+\al+k-m)(\be+m+k)}{(\al+1+k)^2}N^{\al+1,\be-1}_{k},\; k\geq 0.
  \]
\end{prop}
\begin{proof}
  This follows directly from \eqref{eq:Ayjnorm}.
\end{proof}

We summarize the above findings as follows.
\begin{thm}
  Let $m>1$ and $\alpha,\beta>-1$ be such that
  $\alpha>m-2,\sgn(\al-m+1)=\sgn(\be)$ and such that
  \eqref{eq:albenonsing} holds.  Let $\cU$ be the stable, codimension
  $m$ flag spanned by polynomials $y(x)$ such that $(1+x)y'+\be y$ is
  divisible by $\Pud{-\al-1,\be-1}{m}(x)$.  Let $\hW_{\al,\be,m}(x)$
  be the weight defined by \eqref{eq:hWdef}. Then, the $X_m$ Jacobi
  polynomials, as defined by \eqref{eq:hPdef} are the orthogonal
  polynomials obtained by orthogonalizing the flag $\cU$ relative to
  the weight $\hW_{\al,\be,m}(x)$ and subject to the normalization
  condition \eqref{eq:hPnorm}.
\end{thm}

Finally, let us discuss shape-invariance of the generalized Jacobi
operators.  The following Proposition was proved in \cite{GKM10}.
\begin{prop}
  \label{prop:xmjacsi}
  Let $\hA, \hB$ be the operators defined in \eqref{eq:hAabmdef}
  \eqref{eq:hBabmdef}.  Then
  \begin{equation}
    \label{eq:TabmhBhA}
    T_{\al,\be,m} = \hB_{\al,\be,m}\hA_{\al,\be,m},\quad
    \hA_{\al,\be,m} \hPud{\al,\be,m}{m} = 0
  \end{equation}
  is the state-deleting factorization of the $X_m$ Jacobi operator.
  Furthermore,
  \begin{equation}
    \label{eq:Tabmsi}
    T_{\al+1,\be+1,m} = \hA_{\al,\be,m} \hB_{\al,\be,m} + \al+\be+2,
  \end{equation}
  is the dual state-adding factorization.
\end{prop}
\noindent In essence, we are asserting that the generalized operators obey the
same shape-invariance relations as their classical counterparts; c.f.,
equations \eqref{eq:Tabsi1} \eqref{eq:Tabsi2}.

The proof of Proposition \ref{prop:xmjacsi} relies on the
permutability property of higher order Darboux transformation and goes
beyond the scope of this survey.  We limit ourselves to explicitly
deriving the raising and lowering operators used in the above
factorization.

We already know the factorization eigenfunction:
\[ \phi(x) = \hPud{\al,\be,m}{m}(x) = (-1)^{m}
\,\frac{1+a-m}{1+a}\,\Pud{-\al-2,\be}{m} (x). \]
We make use of \eqref{eq:bk2rel} to determine the factorization
gauge.  Making  use of the fact that $h(\al,\be) = (\al+1,\be+1)$ we
obtain
\begin{equation}
  \label{eq:bsi}
  b(x) = \frac{\xi_{\al+2,\be,m}}{\xi_{\al+1,\be-1,m}}.
\end{equation}
We use Proposition \ref{prop:dualbphiW} to derive the dual
factorization gauge
\[ \hb(x) = (1-x^2)\frac{\xi_{\al+1,\be-1,m}}{\xi_{\al+2,\be,m}}\]
and the dual factorization eigenfunction,
\[ \hphi(x) = (1-x)^{-\al-1} (1+x)^{-\be-1} \xi_{\al+1,\be-1,m} .\]
We thereby obtain
\begin{align}
  \label{eq:hAabmdef}
  \hA_{\al,\be,m}(y) &=\frac{\xi_{\al+2,\be,m}}{\xi_{\al+1,\be-1,m}}
  (y'-\rho_{\al+2,\be,m} y)\\
  \label{eq:hBabmdef}
  \hB_{\al,\be,m}(y)
  &=(1-x^2)\frac{\xi_{\al+1,\be-1,m}}{\xi_{\al+2,\be,m}}
  \left[y'-\left(\rho_{\al+1,\be-1,m}+
      \frac{\al+1}{1-x}-\frac{\be+1}{1+x}\right) y\right]
\end{align}
These shape-invariant factorizations serve as a good illustration of
the duality between formal state-adding and state-deleting
transformations; here $\phi(x)$ is a polynomial but $\hphi(x)$ is
merely a quasi-rational function.  

As well, the shape-invariant factorization illustrates that $b(x)$,
the factorization gauge, is not necessarily a polynomial.  Here,
\[ w  = -\phi'/\phi = \rho_{\al+2,\be,m}.\]
The denominator is $\xi_{\al+2,\be,m}$ but the transformation
\[ y\mapsto \xi_{\al+2,\be,m}(y'-w y),\quad y= \hPud{\al,\be,m}{n},\;
n\geq m \]
would produce an imprimitive flag.  The common factor is
$\xi_{\al+1,\be-1,m}$, and that is why the correct factorization gauge
is the rational function shown in \eqref{eq:bsi}.

As a consequence of these shape-invariant factorizations, we have the
following lowering and raising identities for the $X_m$ Jacobi
operators; c.f., \eqref{eq:jacraising} \eqref{eq:jaclower}
\begin{align*}
  &\hB_{\al,\be,m} \hPud{\al+1,\be+1,m}{m+k} = 2(1+k)
  \hPud{\al,\be,m}{m+k+1},\quad k\geq 0;\\
  &\hA_{\al,\be,m} \hPud{\al,\be,m}{m+k} = \frac{1}{2}(1+\al+\be+k)
  \hPud{\al+1,\be+1,m}{m+k-1}.
\end{align*}

\paragraph{\textbf{Acknowledgements}}
\thanks{
We thank Ferenc Tookos for useful comments and suggestions.
The research of DGU was supported in part by MICINN-FEDER grant MTM2009-06973
and CUR-DIUE grant 2009SGR859. The research of NK was supported in part by NSERC
grant RGPIN 105490-2004. The research of RM was supported in part by NSERC grant
RGPIN-228057-2004.


\begin{thebibliography}{99}


\bibitem{Aczel} J. Aczel, Eine Bemerkung \"uber die Charakterisierung
  der klassichen orthogonale Polynome, \emph{Acta
    Math. Acad.Sci. Hungar} \textbf{4} (1953), 315-321.

\bibitem{askey} R.A. Askey and J.A. Wilson, Some basic hypergeometric
  orthogonal polynomials that generalize Jacobi polynomials,
  \emph{Memoirs American Mathematical Society} No. \textbf{319}
  1985.

\bibitem{atkinson} F.V. Atkinson and W.N. Everitt, Orthogonal
  polynomials which satisfy second order differential
  equations. E. B. Christoffel (Aachen/Monschau, 1979), pp. 173--181,
  Birkh\"auser, Basel-Boston, Mass., 1981.

%
\bibitem{bateman}
Erd\'elyi A et al.  \textit{Higher Transcendental Functions, Vol. I},
McGraw-Hill, New York, 1953

\bibitem{Bo} S. Bochner, \"Uber Strum-Liouvillesche Polynomsysteme,
  \emph{Math. Z.} \textbf{29} (1929), 730-736.

\bibitem{CKS95} F. Cooper, A. Khare and U. Sukhatme,
  \textit{Supersymmetry in quantum mechanics} World Scientific,
  Singapore 2001
  
\bibitem{CrooksMilson} P. Crooks and R. Milson, On projective
  equivalence of univariate polynomial subspaces, \textit{SIGMA} 5
  (2009), paper 107

\bibitem{deift} P. A. Deift, Applications of a commutation formula,
  \textit{Duke Math. J.}  \textbf{45}  (1978), 267--310

\bibitem{dubov} S. Y. Dubov, V. M. Eleonskii, and N. E. Kulagin,
  Equidistant spectra of anharmonic oscillators,
  \textit{Sov. Phys.–JETP} \textbf{75} (1992) 47--53

\bibitem{EKLW} W. N. Everitt, K. H. Kwon, L. L. Littlejohn and
  R. Wellman, Orthogonal polynomial solutions of linear ordinary
  differential equations, \emph{J. Comput. Appl. Math} \textbf{133}
  (2001), 85-109.

\bibitem{elw04} W. N.  Everitt, L. L.  Littlejohn, R.  Wellman, The
  Sobolev orthogonality and spectral analysis of the Laguerre
  polynomials ${L\sp {-k}\sb n}$ for positive integers $k$, \textit{J.
    Comput. Appl. Math.}  \textbf{171} (2004), 199--234.

\bibitem{Feldmann} J. Feldmann, On a characterization of classical
  orthogonal polynomials, \emph{Acta Sc. Math.}  \textbf{17} (1956),
  129-133.


\bibitem{GEND83} L.E. Gendenshtein, Derivation of exact spectra of the
  Schrodinger equation by means of supersymmetry, \textit{JETP Lett.}
  \textbf{38} (1983) 356-359\\


\bibitem{gesztesy} F. Gesztesy and G. Teschl, On the double
  commutation method, \textit{Proc. Amer. Math. Soc.} \textbf{124} 1831--1840
  (1996)



\bibitem{GB09} J. Gibbons and A.P. Veselov, On the rational
  monodromy-free potentials with sextic growth, \emph{J. Math. Phys.}
  \textbf{(50)} (2009) 013513,  25 pages.


\bibitem{GKM04a} D. G{\'o}mez-Ullate, N. Kamran and R. Milson, The
  Darboux transformation and algebraic deformations of shape-invariant
  potentials, \emph{J. Phys. A} \textbf{37} (2004), 1789--1804.

\bibitem{GKM04} D. G{\'o}mez-Ullate, N. Kamran and R. Milson,
  Supersymmetry and algebraic Darboux transformations,
  \emph{J. Phys. A} \textbf{37} (2004), 10065--10078.

\bibitem{GKM3} D. G{\'o}mez-Ullate, N. Kamran and R. Milson, \newblock
  Quasi-exact solvability and the direct approach to invariant
  subspaces.  \newblock {\em J. Phys. A}, 38(9):2005--2019, 2005.

\bibitem{GKM6} D. G{\'o}mez-Ullate, N. Kamran and R. Milson, \newblock
  Quasi-exact solvability in a general polynomial setting, \newblock
  {\em Inverse Problems}, \textbf{23} (2007) 1915-1942.

\bibitem{GKM09a} D.~G{\'o}mez-Ullate, N.~Kamran, and R.~Milson,
  \newblock An extension of Bochner's problem: exceptional invariant
  subspaces \textit{J. Approx. Theory} \textbf{162} (2010) 987-1006.

\bibitem{GKM09b} D.~G{\'o}mez-Ullate, N.~Kamran, and R.~Milson, An
  extended class of orthogonal polynomials defined by a
  Sturm-Liouville problem , \textit{J. Math. Anal. Appl.} \textbf{359}
  (2009) 352--367.


\bibitem{GKM10} D.~G{\'o}mez-Ullate, N.~Kamran, and R.~Milson,
  Exceptional orthogonal polynomials and the Darboux transformation,
\textit{J. Phys. A} \textbf{43} (2010) 434016




\bibitem{hendrikssen} E. Hendriksen and H. van Rossum, Semiclassical
  orthogonal polynomials, in ``Orthogonal polynomials and
  applications'' (Bar-le-Duc, 1984), C. Brezinski et al Editors,
  354--361, \emph{Lecture Notes in Math.}, \textbf{1171}, Springer,
  Berlin, 1985.


\bibitem{IsvAs} M. Ismail, \emph{Classical and quantum orthogonal
    polynomials in one variable}, Encyclopedia in Mathematics, vol
  \textbf{98} Cambridge University Press, Cambridge, (2005).


\bibitem{KLS10} R. Koekoek, P. Lesky, and R. Swarttouw, Hypergeometric
  orthogonal polynomials and their q-analogues, Springer-Verlag,
  Berlin, 2010

\bibitem{KL97} K. H. Kwon and L. L. Littlejohn, Classification of
  classical orthogonal polynomials, \emph{J. Korean Math. Soc.}
  \textbf{34} (1997), 973-1008.

\bibitem{Lesky} P. Lesky, Die Charakterisierung der klassischen
  orthogonalen Polynome durch Sturm-Liouvillesche
  Differentialgleichungen,\emph{ Arch. Rat. Mech. Anal.} \textbf{10}
  (1962), 341-352.

\bibitem{Midya-Roy} B. Midya and B. Roy, Exceptional orthogonal
  polynomials and exactly solvable potentials in position dependent
  mass Schr\"odinger Hamiltonians, \textit{Phys. Lett. A}
  \textbf{373(45)} (2009) 4117--4122

\bibitem{Mikolas} M. Mikol\'as, Common characterization of the Jacobi,
  Laguerre and Hermite-like polynomials(in Hungarian),
  \emph{Mate. Lapok} \textbf{7} (1956), 238-248.


\bibitem{Quesne1} C. Quesne, Exceptional orthogonal polynomials,
  exactly solvable potentials and supersymmetry, \textit{J. Phys. A }
  \textbf{41} (2008) 392001--392007

\bibitem{Quesne2} C. Quesne, Solvable rational potentials and
  exceptional orthogonal polynomials in supersymmetric quantum
  mechanics, SIGMA \textbf{5} (2009), 084, 24 pages.

\bibitem{RS} M.~Reed and B.~Simon, {\it Methods of Modern Mathematical
    Physics. II.  Fourier Analysis, Self-Adjointness}, Academic Press,
  New York-London, 1975.

\bibitem{Ron87} A. Ronveaux, Sur l'\'equation diff\'erentielle du
  second ordre satisfaite par une classe de polyn\^omes orthogonaux
  semi-classiques, \emph{C. R. Acad. Sci. Paris S\'er. I Math.}
  \textbf{305} (1987), no. 5, 163--166.


\bibitem{RM89} A. Ronveaux and F.  Marcell\'an, Differential equation
  for classical-type orthogonal polynomials, \emph{
    Canad. Math. Bull.}  \textbf{32} (1989), no. 4, 404--411.

\bibitem{SO1} S. Odake and R. Sasaki, Infinitely many shape
  invariant potentials and new orthogonal polynomials,
  \textit{Phys. Lett. B} \textbf{679} (2009)
  414--417.

\bibitem{SO2} S. Odake and R. Sasaki, Infinitely many shape invariant
  discrete quantum mechanical systems and new exceptional orthogonal
  polynomials related to the Wilson and Askey-Wilson polynomials,
  \textit{Phys. Lett. B} \textbf{682} (2009), 130--136.

\bibitem{SO3} S. Odake and R. Sasaki, Infinitely many shape
  invariant potentials and cubic identities of the Laguerre and Jacobi
  polynomials, \textit{J. Math. Phys.} \textbf{51} (2010), 053513, 9 pages.

\bibitem{SO4} S. Odake and R. Sasaki, Another set of
  infinitely many exceptional ($X_\ell$) Laguerre polynomials,
  \textit{Phys. Lett. B} \textbf{684} (2010),  173--176.


\bibitem{STZ10} R. Sasaki, S. Tsujimoto, and A. Zhedanov, Exceptional
  Laguerre and Jacobi polynomials and the corresponding potentials
  through Darboux-Crum transformations, \textit{J. Phys. A}
  \textbf{43} (2010), 315204, 20 pages


\bibitem{sukumar} C. V. Sukumar, Supersymmetric quantum mechanics of
  one-dimensional systems, {\it J. Phys. A} {\bf 18} (1985) 2917-2936.


\bibitem{Sz} G. Szeg\H{o}, \textit{Orthogonal polynomials},
  Amer. Math. Soc. Colloq. Publ. \textbf{23}, Amer. Math. Soc.,
  Providence RI, 1975, Fourth Edition.

\bibitem{Tanaka} T. Tanaka, $N$-fold Supersymmetry and quasi-solvability associated with $X_2$-Laguerre polynomials, \texttt{ arXiv:0910.0328} [math-ph]

\bibitem{U} V. B. Uvarov, \emph{The connection between systems of polynomials that
are orthogonal with respect to different distribution functions},
USSR Computat. Math. and Math. Phys. \textbf{9} (1969),
25--36.

\bibitem{zettl} A. Zettl, \textit{Sturm-Liouville theory},
  Mathematical Surveys and Monographs, \textbf{121}, Amer. Math.
  Soc., Providence, RI, 2005
\end{thebibliography}
\end{document}